\documentclass[9pt]{IEEEtran}

\usepackage{cite}
\usepackage{amssymb, amsmath, amsthm}
\usepackage[pdftex]{graphicx}
\usepackage{graphicx}
\usepackage{endnotes}
\usepackage[top=54pt,left=54pt,right=54pt,bottom=54pt]{geometry}
\usepackage{amssymb}
\usepackage{epstopdf}
\usepackage{color}
\usepackage[usenames,dvipsnames]{xcolor}
\usepackage{comment}

\newtheorem{theorem}{Theorem}
\newtheorem{definition}{Definition}
\newtheorem{lemma}{Lemma}

\newtheorem{remark}{Remark}

\title{Simultaneous Velocity and Position Estimation via Distance-only Measurements with Application to Multi-Agent System Control}
\author{Bomin Jiang, Mohammad Deghat, \emph{Member, IEEE} and \\
Brian D.O. Anderson, \emph{Life Fellow, IEEE}
\thanks {*This work is supported by National ICT Australia, which is funded by the Australian Research Council through the ICT Centre of Excellence program.}
\thanks{Bomin Jiang and Brian D. O. Anderson are with Research School of Engineering, The Australian National Universityand National ICT Australia. Mohammad Deghat is with the School of Engineering and Information Technology, the University of New South Wales and used to be with the ANU and NICTA. \tt\small{(e-mail: \{U5225976,Brian.Anderson\}@anu.edu.au}, \tt\small{m.deghat@unsw.edu.au)}}
}

\begin{document}
\maketitle
\thispagestyle{empty}
\pagestyle{empty}

\begin{abstract}
This paper proposes a strategy to estimate the velocity and position of neighbor agents using distance measurements only. Since with agents executing arbitrary motions, instantaneous distance-only measurements cannot provide enough information for our objectives, we postulate that agents engage in a combination of circular motion and linear motion. The proposed estimator can be used to develop control algorithms where only distance measurements are available to each agent. As an example, we show how this estimation method can be used to control the formation shape and velocity of the agents in a multi agent system. Simulation results are provided to illustrate the performance of the proposed algorithm.
\end{abstract}

\section{Introduction}

The performance of multi-agent systems in various tasks, e.g. consensus \cite{ren2007information,olfati2007consensus}, formation shape control \cite{anderson2008rigid,ren2010distributed}, cooperative geolocalization \cite{wymeersch2009cooperative}, etc. has been studied with increasing intensity over recent years.
These tasks are usually required to be performed in a decentralised way \cite{foderaro2010necessary} and using limited information, i.e. each agent should individually identify possible actions, and while such actions are required to achieve the final goal of the formation, each agent can communicate only with its neighboring agents.
Examples of these tasks are retrieving information from an area covered by a sensor network (where the agents are sensors deployed in the area), or moving together in a desired formation shape from one point to another where the agents are ground or aerial vehicles.

In a formation control problem, which is the focus of this paper, each agent tries to contribute to achieve the global goal of the formation using measurements of, typically, relative position and velocity of its neighbors. Examples of such problems are given in \cite{anderson2011range,Krick,Summers11,Dimos12}.
These problems become more challenging when the agents cannot instantaneously measure all the information required to apply motion corrections to achieve the final goal of the formation and have to estimate some of this information using their measurements.

An example of such a challenging problem is given in \cite{cao2011formation}, where a formation (shape and translation motion) control method, called stop-and-go, has been devised to control the agents not able to measure the relative positions (both distance and angle) of their neighbors, but only able to measure the distances to their respective neighbors. This measurement restriction makes the control problem significantly harder. Beyond that, this paper makes an assumption that the formation has a leader agent whose velocity is constant, and the followers take up positions while moving with the same velocity as the leader.

This paper treats a related problem. Agents are required to estimate the relative position and velocity of their neighbor agents using only distance measurements to the neighbors, and achieve both velocity consensus and formation shape control.
The key is to postulate that the motion of each agent comprises two parts: a translation and a circular motion. The circular motion is around a moving center, and it is the centers of each agent's motion, rather than the agents themselves, which achieve velocity consensus.
The purpose of the superimposed circular motion is to allow inter-agent localization and velocity estimation, not using instantaneous measurements, but using distance measurements collected over an interval. We postulate that neighbor agents remain in communication even if they initially have different velocities.

The notion of using deliberate motions of agents to assist in localization was suggested in \cite{zhu2011sensor}, in relation to sensor network localization. The idea in \cite{zhu2011sensor} is
that if each node in a sensor network moves in a small neighbourhood of its original position, it is possible to infer direction information from distance measurements.
Our idea is similar; however, the motions in \cite{zhu2011sensor} are random while this paper studies the localization problem using distance-only measurements when agents are executing independent circular motions and it further discusses the situation where agents are performing a combination of circular motion and linear motion, with the linear motion components required to achieve velocity consensus.
In addition, the idea of introducing sinusoidal perturbation in formation control problems is not wholly novel: in \cite{tian2013global}, the authors have introduced sinusoidal perturbations to the usual gradient based control algorithm in order to achieve a different objective.
An advantage of having a combination of linear and circular motion over only linear motion as in \cite{cao2011formation} is that the agents are less likely to travel out of communication range during the localization process.


An abbreviated conference version of this paper has been presented in \cite{6760298}.
The novel contributions of the paper, in comparison to the conference paper \cite{6760298}, are as follows (a) proposing a discrete time control algorithm to achieve velocity consensus and simultaneous formation shape control, which can be used when distance-only measurements are available, and (b) introducing an improvement by adaptively adjusting the circular motion radius.

The rest of this paper is organised as follows. Section~\ref{sec:estimation_translational} gives a solution to the location and velocity estimation problem using distance-only measurements when each agent is executing a combination of linear motion and circular motions. Section~\ref{sec:adaptive} discusses an improvement of the algorithm derived in the previous section involving an adaptively adjusting the circular motion radius. Section~\ref{sec:control} discusses a discrete time control algorithm to achieve velocity consensus and formation shape control with distance-only measurements. Simulation is included in each section. Concluding remarks and directions for future research are given in Section~\ref{sec:Conclusion}.

\section{Relative position and velocity estimation using sinusoidal perturbation}
\label{sec:estimation_translational}

In the conference version of this paper \cite{6760298}, we gave detailed explanation on how to infer neighbouring agents' relative position and velocity. We have also discussed special cases very carefully. Here we just give a brief introduction to the ideas.



\subsection{Problem statement}
\label{sec:problemStatement}
Consider two point agents, 1 and 2. Each agent performs a combination of circular and rectilinear motion, so each has a certain radius, direction and angular velocity for the circular motion and velocity for the rectilinear motion. Agent 1 knows its own radius, angular velocity and the translational velocity of its circle centre and can only measure (continuously) the distance but not bearing of agent 2. Conversely, agent 2 knows its radius, angular velocity and the velocity of its circle centre and can only measure the distance of agent 1. The goal is for both agents to localize and sense the velocities of each other for velocity consensus purposes.

\begin{figure}[h]
\begin{center}
\includegraphics[width=5.5cm]{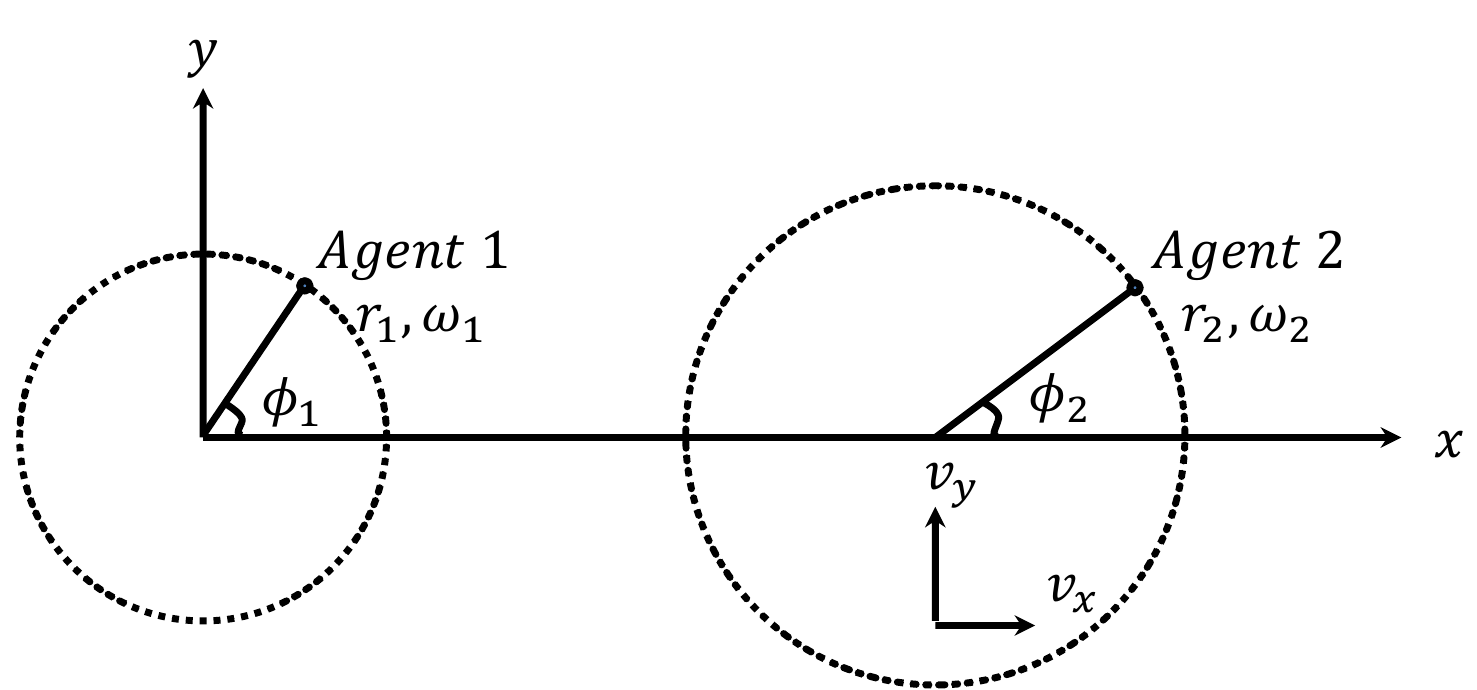}
\end{center}
\caption{Set up a coordinate system with respect to agent 1's circle centre}
\label{velocity consensus 1}
\end{figure}

As shown in Fig.~\ref{velocity consensus 1}, we set up (for analysis purposes by us) a global coordinate system with origin at agent 1's circle centre and agent 2's circle centre on the $x$ axis when $t=0$. Suppose $r_i$ is the radius of agent $i$'s motion, $\omega_i$ is the angular velocity of agent $i$, $z(t)$ is the distance at time $t$ between agent 1 and 2 and $d$ is the distance between the two circle centres. The coordinate system is defined by the agent pair, and is used for analysis purposes by us. Its orientation with respect to agent 1's local coordinate basis is not known by agent 1 at this stage though the orientation can be obtained after that agent learns $\phi_1$.
In addition, let $v_{c_i}$ be the velocity of agent $i$'s circle centre, $v_{ji}$ be the relative velocity of agent $j$'s circle centre with respect to agent $i$'s circle centre, $v_x$ be the $x$ component of the velocity $v_{21}$, and $v_y$ be the $y$ component of the velocity $v_{21}$.
The positive direction of angular velocities is counter-clockwise.

We assume in this paper that $v_x$ and $v_y$ are constant for $kT<t<(k+1)T$, $T>0$, $k=0,1,2,\cdots$ and may only change at time instants $kT$, perhaps reflecting a discrete-time consensus algorithm. We explain later how to choose $T$.

There holds
\begin{equation}\label{velocity equation}
\begin{split}
z^2(t)=&[d+v_xt+r_2cos({\omega_2t+\phi_2})-r_1cos {(\omega_1t+\phi_1)}]^2\\
&+[v_yt+r_2sin{(\omega_2t+\phi_2)}-r_1sin{(\omega_1t+\phi_1)}]^2\\
\end{split}
\end{equation}

Let $d_x=d+v_xt$ and $d_y=v_yt$ and rewrite \eqref{velocity equation} using easy algebra as:
\begin{equation}\label{rewrite}
\begin{split}
z^2(t)=&(d_x^2+d_y^2+r_1^2+r_2^2)\\
&+2d_xr_2cos({\omega_2t+\phi_2})+2d_yr_2sin({\omega_2t+\phi_2})\\
&-2d_xr_1cos {(\omega_1t+\phi_1)}-2d_yr_1sin {(\omega_1t+\phi_1)}\\
&-2r_1r_2cos[(\omega_1-\omega_2)t+(\phi_1-\phi_2)]\\
\end{split}
\end{equation}

\subsection{Finding relative position and velocity of a neighbour}

In the system comprising a pair of agents 1 and 2, without loss of generality, we only show how agent 1 can localize and estimate the relative velocity of agent 2. The first step is for agent 1 to identify the angular velocity of agent 2, using a Fourier representation of $z^2(t)$ for $t\in[0,T]$. When $\|v_{21}\|$ is sufficient small in comparison to $r_1$, $r_2$ and $d$, four distinct peaks will show up at 0, $|\omega_1|$, $|\omega_2|$ and $|\omega_1-\omega_2|$ in frequency domain. This allows agent 1 to pick up the angular velocity of agent 2. More insights about this assumption will be discussed in Section \ref{sec:adaptive}; detailed explanation about how to identify the angular velocity is given in \cite{6760298}.

In order to identify the value of $d$, $\phi_1$, $v_x$ and $v_y$, we allow agent 1 to measure the distance between the two agents $z(t)$ and analyse the Fourier series of the periodic extension of $z^2(t)$. Lemmas \ref{lemma show the Fourier of linear times sinusoid} and \ref{lemma show the Fourier of quatratic} show the Fourier series of some summands arising in \eqref{rewrite} and Lemma \ref{lemma full rank} will provide the tool to show that these summands are linearly independent and can be identified separately. Theorem \ref{process} gives details of the procedure to identify $d$, $\phi_1$, $v_x$ and $v_y$.

\begin{lemma}
\label{lemma show the Fourier of linear times sinusoid}
Suppose $d$, $v$, $\omega_1$ and $T$ are positive constants with $T=k_1\frac{2\pi}{\omega_1}$ for some positive integer $k_1$ and $f_1(t)=(d+vt)cos(\omega_1 t)$ $\forall t\in [0,T]$. Define $f_1'(t)$ to be the periodic extension of $f_1(t)$ such that $\forall t\in[0,T),f_1'(t)=f_1(t)$ and $\forall t\in(-\infty,\infty),f_1'(t)=f_1'(t+T)$. Let $c_n$ be the coefficients of its Fourier series $f_1'(t)=\sum_{n=-\infty}^\infty c_ne^{j2\pi nt/T}$
Then, if $n=k_1$, there holds
\begin{equation}
\label{Fourier k_1}
c_{k_1}=\frac{1}{2}(d+\frac{1}{2}vT)+\frac{vT}{8\pi k_1}j
\end{equation}
and if $n\neq k_1$ and $n>0$ there holds
\begin{equation}
\label{Fourier others}
c_n=\frac{vTj}{4\pi}(\frac{1}{n-k_1}+\frac{1}{n+k_1})
\end{equation}
\end{lemma}

\begin{proof}
The lemma above can be proved in a straightforward manner by calculating the value of
\begin{equation}
\label{Fourier series}
c_n=\frac{1}{T}\int^T_0 (d+vt)cos(\omega_1 t)e^{-jn\frac{2\pi}{T}t}dt
\end{equation}
\end{proof}

\begin{lemma}
\label{lemma show the Fourier of quatratic}
Suppose $a$, $b$ and $T$ are positive constants (with $T$ not necessarily a multiple of $2\pi/\omega_1$). Define $f_2(t)=at^2+bt$ $\forall t\in [0,T]$. Note the domain of definition of $f_2(t)$ is bounded. Define $f_2'(t)$ to be the periodic extension of $f_2(t)$ such that $\forall t\in[0,T),f_2'(t)=f_2(t)$ and $\forall t\in(-\infty,\infty),f_2'(t)=f_2'(t+T)$. Let $c_n$ be the coefficients of the Fourier series $$f_2'(t)=\sum_{n=-\infty}^\infty c_ne^{i2\pi nt/T}$$
Then, for all $n=\pm1,\pm2,\cdots$, there holds $c_{n}=\frac{aT^2}{2\pi^2n^2}+\frac{aT^2+Tb}{2\pi n}j$ ($j^2=-1$)
\end{lemma}

\begin{proof}
The lemma can be proved in a straightforward manner by calculating the value of
\begin{equation}
\label{at^2}
c_n=\frac{1}{T}\int^T_0 (at^2+bt)e^{-jn\frac{2\pi}{T}t}dt
\end{equation}
\end{proof}

\begin{lemma}
\label{lemma full rank}
Suppose $n_1$, $n_2$, $n_3$, $n_4$, $k_1$ and $k_2$ are six different positive integers. Then the matrix
\begin{equation}
\label{linear find matrix}
\left[
       \begin{array}{cccc}
         \frac{1}{n_1^2} & \frac{1}{n_1} & \frac{1}{n_1-k_1}+\frac{1}{n_1+k_1}& \frac{1}{n_1-k_2}+\frac{1}{n_1+k_2}\\
         \frac{1}{n_2^2} & \frac{1}{n_2} & \frac{1}{n_2-k_1}+\frac{1}{n_2+k_1}& \frac{1}{n_2-k_2}+\frac{1}{n_2+k_2}\\
         \frac{1}{n_3^2} & \frac{1}{n_3} & \frac{1}{n_3-k_1}+\frac{1}{n_3+k_1}& \frac{1}{n_3-k_2}+\frac{1}{n_3+k_2}\\
         \frac{1}{n_4^2} & \frac{1}{n_4} & \frac{1}{n_4-k_1}+\frac{1}{n_4+k_1}& \frac{1}{n_4-k_2}+\frac{1}{n_4+k_2}\\
       \end{array}
\right]
\end{equation}
is full rank.
\end{lemma}
\begin{proof}
The lemma can be proved in a straightforward manner by calculating the value of the matrix determinant.
\end{proof}

In the following theorem, we show that each agent can estimate the position and translational velocity of the other agent using distance-only measurements over an interval of time $T$.
For now we assume that the angular velocities of agents 1 and 2 are commensurate. We later explain what happens if $\omega_1$ and $\omega_2$ are incommensurate.
\begin{theorem}
\label{process}
For a pair of point agents in $\mathbb{R}^2$, if each agent is executing a combination of circular motion and linear motion and the associated angular frequencies are commensurate, each agent can find the position and translational velocity of the other agent by distance-only measurements over an interval.
\end{theorem}

\begin{proof}
The definitions of $r_1$, $r_2$, $\omega_1$, $\omega_2$, $d$, $z$, $v_x$, $v_y$, $\phi_1$ and $\phi_2$ are the same as in Section \ref{sec:problemStatement}. We choose $T$ so that there exist integers $k_1,k_2$ defining the multiple which $T$ represents of the periods associated with the two angular velocities, i.e. $k_1=\frac{\omega_1T}{2\pi}$ and $k_2=\frac{\omega_2T}{2\pi}$. The existence of $k_1$ and $k_2$ relates to the concept of commensurable numbers, see Remark 3 in \cite{6760298}.

Suppose one continuously measures $z$ for a time period $T$ and finds the Fourier series of the periodic extension of $z^2$. Consider \eqref{rewrite} and suppose $c_n$ are the coefficients of Fourier series of the periodic extension of $z^2$, $s_n$ are the coefficients of Fourier series of the periodic extension of $(d_x^2+d_y^2+r_1^2+r_2^2)$, $u_n$ are the coefficients of Fourier series of the periodic extension of $-2d_xr_1cos(\omega_1t+\phi_1)-2d_yr_1sin(\omega_1t+\phi_1)$ and $w_n$ are the coefficients of Fourier series of the periodic extension of $2d_xr_2cos(\omega_2t+\phi_2)+2d_yr_2sin(\omega_2t+\phi_2)$.

From \eqref{rewrite} we know that for any $n>0\cap n\neq |k_1-k_2|$ there holds
\begin{equation}
\label{linear combination1}
c_n=s_n+u_n+w_n
\end{equation}
Note the coefficients of the Fourier series of the term $-2r_1r_2cos[(\omega_1-\omega_2)t+(\phi_1-\phi_2)]$ in \eqref{rewrite} are all zero except for the index $n=|k_1-k_2|$.

Define constants:
\begin{equation}
\label{definition of U}
U=r_1(j\frac{v_xT}{4\pi}+\frac{v_yT}{4\pi})e^{j(\phi_1+\pi)}
\end{equation}
and
\begin{equation}
\label{definition of W}
W=r_2(j\frac{v_xT}{4\pi}+\frac{v_yT}{4\pi})e^{j(\phi_2)}
\end{equation}
Suppose further that
\begin{equation}
\label{definition of R}
R=\frac{(v_x^2+v_y^2)T^2}{2\pi^2}
\end{equation}
and
\begin{equation}
\label{definition of I}
I=\frac{(v_x^2+v_y^2)T^2+2v_xdT}{2\pi}
\end{equation}
From Lemma \ref{lemma show the Fourier of linear times sinusoid}, Lemma \ref{lemma show the Fourier of quatratic} and \eqref{rewrite} we know for any $n>0\cap n\neq k_1,k_2~or~|k_1-k_2|$ there holds

\begin{equation}
\label{linear combination2}
\begin{split}
c_n&=\frac{1}{n^2}R+\frac{1}{n}Ij\\
&+(\frac{1}{n-k_1}+\frac{1}{n+k_1})\cdot 2U+(\frac{1}{n-k_2}+\frac{1}{n+k_2})\cdot 2W
\end{split}
\end{equation}

From \eqref{linear combination2} and Lemma \ref{lemma full rank} we know that if we have four values of $c_n,n>0\cap n\neq k_1,k_2~or~|k_1-k_2|$, we are able to find the unique solutions of $R$, $I$, $U$ and $W$.
Because $d$, $v_x$ and $v_y$ are all real numbers, ideally $R$ and $I$ should also be real numbers. Sometimes due to noise or error, the $R$ and $I$ obtained from matrix operations may be complex numbers, but this will not affect the process below.

Now we have the value of $U$ and $W$ and can obtain $u_{k_1}$. Furthermore, from Lemma \ref{lemma show the Fourier of linear times sinusoid} and \eqref{definition of U} we know that
\begin{equation}\label{ck1}
u_{k_1}-\frac{U}{2k_1}=\frac{r_1}{2}(d+\frac{1}{2}v_xT-\frac{1}{2}v_yTj)e^{j(\phi_1+\pi)}
\end{equation}
\begin{equation}\label{ck1+1}
U=r_1(j\frac{v_xT}{4\pi}+\frac{v_yT}{4\pi})e^{j(\phi_1+\pi)}
\end{equation}
and $d$, $v_x$, $v_y$ and $\phi_1$ can be found from these equations.

The solutions for $d$ and $\phi_1$ are given by
\begin{equation}\label{d}
d=\frac{2}{r_1}|u_{k_1}-\frac{U}{2k_1}+\pi jU|
\end{equation}
\begin{equation}\label{phi}
\phi_1=arg(u_{k_1}-\frac{U}{2k_1}+\pi jU)+\pi
\end{equation}
and the solutions for $v_x$ and $x_y$ are given by
\begin{equation}\label{vy}
v_x=\mathop{\rm Im}(\frac{4\pi U}{T r_1 e^{j(\phi_1+\pi)}})
\end{equation}
\begin{equation}\label{vx}
v_y=\mathop{\rm Re}(\frac{4\pi U}{T r_1 e^{j(\phi_1+\pi)}})
\end{equation}
\end{proof}

\begin{remark}
\label{rem:no rotation}
In the special situation where there are no rotations and both agents are executing linear motion, the absolute value of the relative velocity and distance between these agents can be obtained from the Fourier series of the term $(d_x^2+d_y^2+r_1^2+r_2^2)$ in \eqref{rewrite} but the direction cannot be found. This result is the same as the situation described in Section 5.2.1 of the previous paper \cite{cao2011formation}.
\end{remark}

\begin{remark}
\label{incommensurate}
When $\omega_1$ and $\omega_2$ are incommensurate\footnote[1]{Two non-zero real numbers $a$ and $b$ are said to be commensurable if $a/b$ is a rational number.}, then $z^2(t)$ in \eqref{rewrite} is an almost periodic function \cite{corduneanu1989almost} and one cannot have $T=k_1\frac{2\pi}{\omega_1}=k_2\frac{2\pi}{\omega_2}$ with $k_1/k_2$ a rational number. Thus at least one or maybe both of $k_1$ and $k_2$ are not integers. Now $T$ should be chosen (and, as guaranteed by the theory of almost periodic functions, it can be so chosen by taking it sufficiently large) to ensure that both $k_1,k_2$ are close to integers (and indeed one may be an integer). Then the Fourier coefficients in Lemma~\ref{lemma show the Fourier of linear times sinusoid} and Lemma~\ref{lemma show the Fourier of quatratic} are different; their expressions have extra additive terms which are small if the deviation of $k_1$ and $k_2$ from integer numbers are small. Thus in Theorem~\ref{process} we can still find $d$, $v_x$, $v_y$ and $\phi_1$ with some error which is also small if the deviations of $k_1$ and $k_2$ from integer numbers are small. The longer $T$ is, the more accurate the results are.
\end{remark}

\section{Assigning Radius Adaptively}
\label{sec:adaptive}
In the above sections, we let each agent infer the position and relative velocity information of neighbouring agents by 1) carrying out a Fourier transform and then 2) identifying peaks to estimate $\omega$ of neighbouring agents 3) solving the set of linear equations \eqref{linear combination2}. In step 2) if $\|v_{12}\|T$ is sufficiently small in comparison to $r_1,r_2$, we can show that there are always peaks at $k_1$ and $k_2$.

\begin{lemma}
\label{lem::have peak}
Adopt the hypothesis in Theorem \ref{process} and consider \eqref{linear combination1}. There exists a positive real number $\alpha$ such that if $d>r_1,r_2>\alpha \cdot \|v_{12}\|$, then  $c_n$ (regarded as a function of the integer $n$) has peaks at $n=k_1$ and $n=k_2$.
\end{lemma}

\begin{proof}
When $n\neq 0,k_1,k_2,|k_1-k_2|$, the dependence of $\|c_n\|$ on $\|v_{12}\|$ can be expressed as follows
\begin{equation}\label{normal cn}
\begin{split}
\|c_n\|&=\Big\|h_1(k_1,k_2,T)\|v_{12}\|^2+ h_2(k_1,k_2,T) \|v_{12}\| d\\ &+h_3(k_1,k_2,T)\|v_{12}\| r_1+h_4(k_1,k_2,T)\|v_{12}\| r_2\Big\|.
\end{split}
\end{equation}
On the other hand, when $n=k_1$, there holds
\begin{equation}\label{special cn}
\begin{split}
\|c_n\|&=\Big\|h_5(k_1,k_2,T)\|v_{12}\|^2+ h_6(k_1,k_2,T) \|v_{12}\| d \\ &+h_7(k_1,k_2,T)\|v_{12}\| r_2+h_8(k_1,k_2,T)d r_1\Big\|
\end{split}
\end{equation}
where $h_{1,2,\cdots,8}(k_1,k_2,T)$ are all bounded functions for integer $k_1,k_2$ and $T>0$.

Now compare \eqref{normal cn} and \eqref{special cn}. For a large enough $\alpha$ the term $h_8(k_1,k_2,T)d r_1$ will be dominant and thus there will be peaks recognised at $n=k_1$. Similarly there will also be peaks recognised at $n=k_2$.
\end{proof}

According to Lemma \ref{lem::have peak}, the proposed algorithm works if $d>r_1,r_2>\alpha \cdot \|v_{12}\|$. In reality, $d>r_1,r_2$ is automatically satisfied if we aim to avoid collision. Furthermore, in order to ensure that $r_1,r_2>\alpha \cdot \|v_{12}\|$ holds for each agent pair, we propose an adaptive radius algorithm whereby $r_i$ of each agent re-set at the end of each $T$ second intervals as follows
\begin{equation}
r_i\Big((k+1)T\Big)=\alpha \cdot \max_j\{ \|v_{ij}(kT)\| \}
\end{equation}
where $j$ denotes the indeces of neighbouring agents of $i$ and $\alpha$ is a sufficient large value. Note that because $r_i$ only changes at the end of each interval $T$, the radius is fixed within each interval.

The adaptive radius law will ensure that $r_1,r_2>\alpha \cdot \|v_{12}\|$ holds for each agent pair. Furthermore, as velocity consensus is being achieved, $\|v_{ij}(kT)\|$ will approach zero and so will $r_i$. It is noticeable that the accuracy of estimation of $\|v_{ij}(kT)\|$ is independent of the value of radius. Each agent can estimate the \textit{norm} of velocities of neighbours' circle centres via $R$ according to \eqref{definition of R}, even if no peaks are identified. This phenomenon is consistent with the paper \cite{cao2011formation}, which shows that without circular motions, for agents only doing linear motions, it is possible to estimate the norm of relative velocities of neighbours, even though the directions are left unknown.

When there are sudden changes in velocities of agents due to e.g. wind or deliberate change of course by a leader agent, an already achieved consensus and formation may be broken. In this case, even if the radius of the circle of each agent has already approached to zero, each agent can still obtain a good estimate of the absolute value of velocities of its neighbours' circle centres. This can result in an increase of radius of circular motions in response to the broken consensus, which allows the agents to achieve velocity consensus and formation shape control again.

A simple demonstration of this idea is shown in the figure below. Similarly to the setting in Section V-B in \cite{6760298}, consider a multi-agent system shown in Fig. \ref{Fig1}, suppose $\omega_i$ is the angular velocity of agent $i$, $T$ is the sampling time interval, $(v_{xi},v_{yi})$ is the translational velocity of agent $i$ and $(p_{xi},p_{yi})$ is the position of circle centre of agent $i$. In the simulation, we set $\omega_1=\omega_3=5$, $\omega_2=-3$, $T=2\pi$. When $t=0$, $(v_{x1},v_{y1})=(-4,2)$, $(v_{x2},v_{y2})=(3,-2)$, $(v_{x3},v_{y3})=(2,4)$, $(p_{x1},p_{y1})=(70,30)$, $(p_{x2},p_{y2})=(0,50)$ ,$(p_{x3},p_{y3})=(0,0)$ and $\varepsilon=0.35$. All the parameters are in SI units.

\begin{figure}[h]
\begin{center}
\includegraphics[width=0.15\paperwidth]{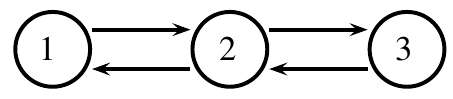}
\end{center}
\caption{A example of multi-agent system}
\label{Fig1}
\end{figure}

Fig. \ref{Re-achieved consensus_1} and Fig. \ref{Re-achieved consensus_2} show a simulation where there is a sudden change in velocity of agent 2 at $t=20T$. In Fig. \ref{Re-achieved consensus_1}, it is clear that the radius of circular motions of agents increases in response to the broken consensus. In Fig. \ref{Re-achieved consensus_2}, it is shown that the consensus of translational velocities is re-achieved after this sudden change.

\begin{figure}[h]
\begin{center}
\includegraphics[width=0.35\paperwidth]{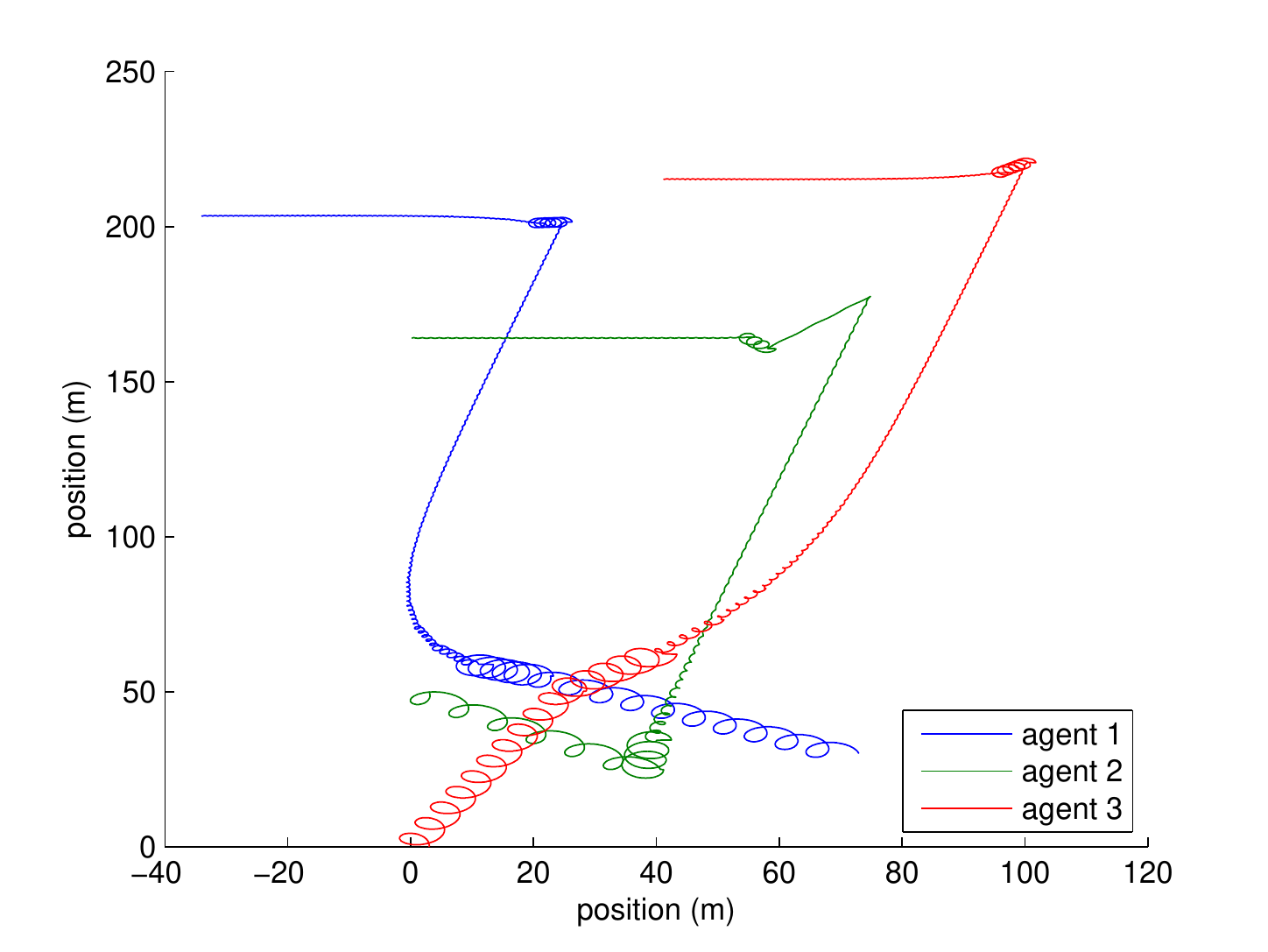}
\end{center}
\caption{Re-achieved consensus}
\label{Re-achieved consensus_1}
\end{figure}

\begin{figure}[h]
\begin{center}
\includegraphics[width=0.4\paperwidth]{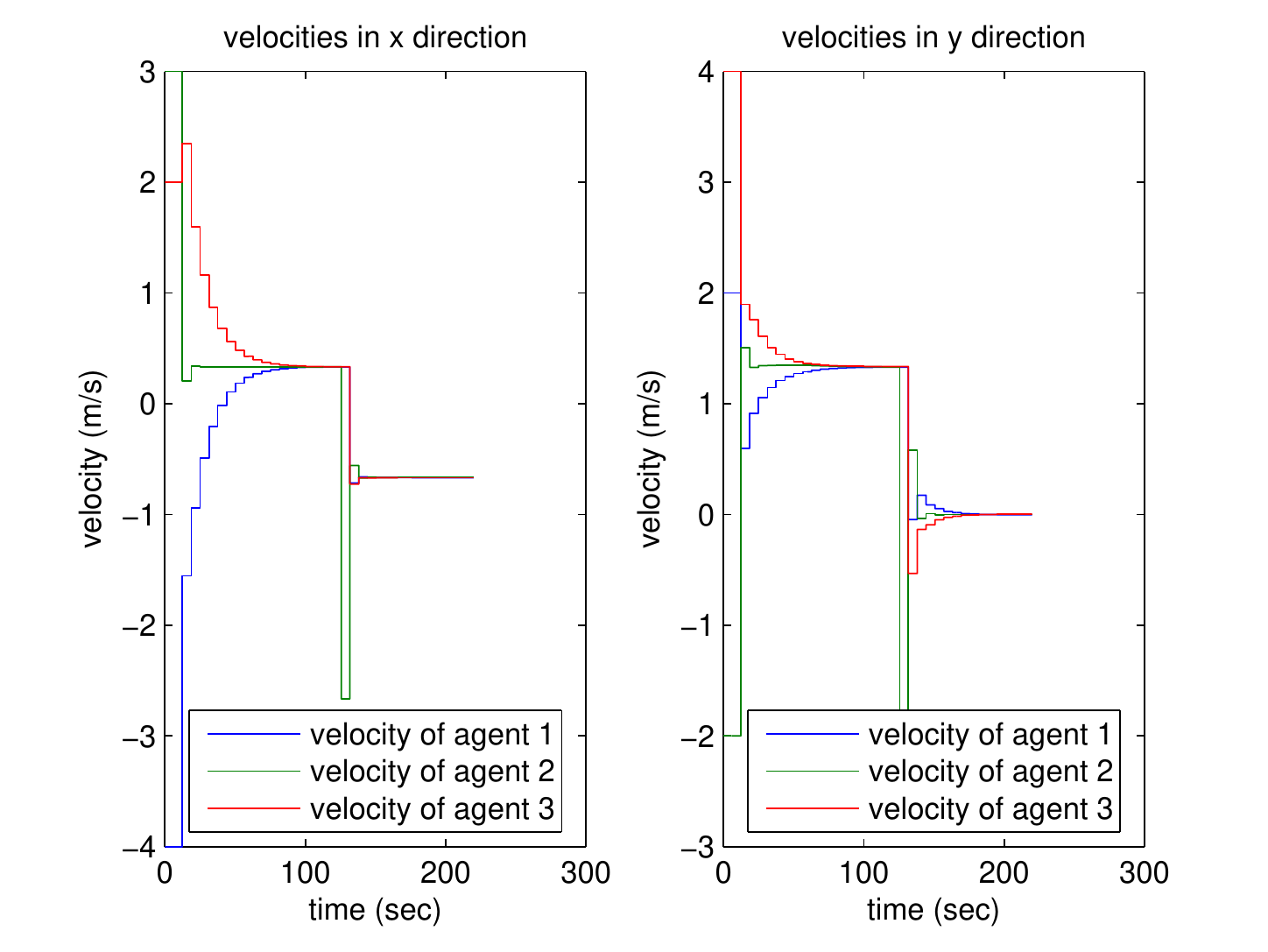}
\end{center}
\caption{Velocities of circle centres in re-achieved consensus}
\label{Re-achieved consensus_2}
\end{figure}

\section{Combining Velocity Consensus and Formation Shape Control}
\label{sec:control}
\subsection{Stability of discrete time control algorithm}
\label{sec::Stability of discrete time control algorithm}
In \cite{anderson2012combining}, an algorithm is derived to combine velocity consensus with formation shape control. The algorithm deals with a continuous time problem and is in the form
\begin{equation}
\label{formation conti}
\begin{split}
&\dot{p}_i=v_i\\
&\dot{v}_i=\sum_{j\in \mathcal{N}_i}(v_j-v_i)+2\sum_{j\in \mathcal{N}_i}(d^{*2}_{ij}-d^{2}_{ij})(p_i-p_j)
\end{split}
\end{equation}
where $p_i$ is the position of the $i$th agent, $i=1, \cdots, N$, $v_i$ is the velocity of the $i$th agent and $\mathcal{N}_i$ is the set of neighbouring agents of agent $i$. Further, $d^{*}_{ij}$ is the desired distance between agent $i$ and $j$ and $d_{ij}$ is the current distance between agent $i$ and $j$. In our context, agent positions and velocities refer to the centre of the circular motion.
The system equations \eqref{formation conti} can be written in the matrix form
\begin{equation}\label{eq:system_1}
\begin{split}
 \dot p &= v \\
 \dot v &= -(\mathcal{L}\otimes I_2)v + f(p)
\end{split}
\end{equation}
where $p\in\mathbb{R}^{2N}$ is the vector of all $p_i$ stacked together, $\mathcal{L}$ denotes the Laplacian matrix which is positive semi-definite and has one zero eigenvalue when the graph is connected and undirected, and $f(p)$ is a vector with the entries $2\sum_{j\in \mathcal{N}_i}(d^{*2}_{ij}-d^{2}_{ij})(p_i-p_j)$, $i=1, \cdots, N$.

In our case, the algorithm cannot be implemented directly because we let each agent measure distance for a time period $T$ and then make a velocity adjustment at the end of each such interval. A discrete version of \eqref{formation conti} for our use is given by
\begin{small}
\begin{equation}
\label{formation discrete}
\begin{split}
&\dot{p}_{i}=v_{i}\\
&v_{i}\big((k+1)T\big)=v_{i}(kT) +\epsilon_1 T\sum_{j\in \mathcal{N}_i}\Big(v_{j}(kT)-v_{i}(kT)\Big)\\
&~~~~~~~~+2\epsilon_2 T\sum_{j\in \mathcal{N}_i}\Big(d^{*2}_{ij}-d^{2}_{ij}(kT)\Big)\Big(p_i(kT)-p_j(kT)\Big)
\end{split}
\end{equation}
\end{small}
where $\epsilon_1, \epsilon_2$ are suitably small positive constants; more information is given below. Note that the first equation remains in continuous time while the second equation is discretised. However, since $v_{i}(t)$ is constant over an interval $T$, it follows that the discretisation of the first equation, viz. $p_{i}\big((k+1)T\big)=p_{i}(kT)+Tv_{i}(kT)$ exactly interpolates the continuous function $p_{i}(t)$ for $t=kT$ with integer $k$.

To show the stability of \eqref{formation discrete}, we start with the continuous-time system and make the following transformation
\begin{equation}\label{eq::transform p p bar}
\bar p_r = Rp,~~~\bar v_r = Rv
\end{equation}
where $R$ is an orthonormal matrix whose first two rows are $(\boldsymbol{1}\otimes I_2)^\top/\sqrt{N}$, $\bar{p}_r:=[p_0^\top~\bar{p}^\top]^\top$ with $p_0\in\mathbb{R}^2$ and $v:=[v_0^\top~\bar{v}^\top]^\top$ with $v_0\in\mathbb{R}^2$. Then $\dot{v}_0=0$, that is the position of the center of mass of the agents in $\bar p$-coordinates is constant, and the system equations in $\bar p$ and $\bar v$ are
\begin{equation}\label{eq:system_2}
\begin{split}
 \dot{\bar{p}} &= \bar v \\
 \dot{\bar{v}} &= L\bar{v} + \bar{f}(\bar{p})
\end{split}
\end{equation}
where $L$ is the $(2N-2)\times(2N-2)$ nonzero block of $-R(\mathcal L \otimes I_2)R^\top$ which is negative definite and $\bar f(\bar p)$ contains the nonzero entries of $R f(R^\top \bar p)=R f(p)$.

Our approach to show the stability of \eqref{formation discrete} is as follows: first we define a Malkin structure in Definition \ref{def::malkin strucure}. After that we show in Lemma~\ref{lemma transfer Malkin} that \eqref{eq:system_2} has Malkin structure. Then we develop in Theorem~\ref{lemma dirscrete Malkin} a discrete-time version of the continuous-time Malkin's theorem as invoked by Krick  \cite{Krick2007formation}. Finally, we use these results and show in Theorem~\ref{theorem_main_results} that  \eqref{formation discrete} is stable for sufficiently small values of $\epsilon_1$ and $\epsilon_2$.

\begin{definition}[Malkin structure]\footnote[1]{There are minor differences in the definition of Malkin structure in different references. We use the definition in \cite{Krick2007formation} here.}
\label{def::malkin strucure}
A system has Malkin structure if it is in the form
\begin{equation}
\label{eq:malkin_structure}
\dot r=\left[\begin{array}{cc}
            0 & 0 \\
            0 & A
          \end{array}
\right]r+g(\theta,\rho),
~
r=\left[\begin{array}{c}
                         \theta \\
                         \rho
                       \end{array}
\right],~
g=\left[\begin{array}{c}
                         \Theta(\theta,\rho) \\
                         P(\theta,\rho)
                       \end{array}
\right]
\end{equation}
where $A$ has eigenvalues with negative real parts. Furthermore, $g(\theta,\rho)$ is a second order term with the following conditions i) $g(\theta,0)=0$, ii) there exists
$$h_1(\theta)=\lim_{\rho\rightarrow 0}\frac{\Theta(\theta,\rho)}{\|\rho\|},~~~~~~h_2(\theta)=\lim_{\rho\rightarrow 0}\frac{P(\theta,\rho)}{\|\rho\|},$$
$$b_1=\left\{\begin{array}{c}
               \frac{\Theta(\theta,\rho)}{\|\rho\|} ~if~\rho\neq 0 \\
               h_1(\theta)~~~ if~\rho=0
             \end{array}
\right.,~
b_2=\left\{\begin{array}{c}
               \frac{P(\theta,\rho)}{\|\rho\|} ~if~\rho\neq 0 \\
               h_2(\theta)~~~ if~\rho=0
             \end{array}
\right.$$
such that $b_1$ and $b_2$ are bounded smooth functions and $b_2(0)=0$;
\end{definition}

\begin{lemma}
\label{lemma transfer Malkin}
The system equations in \eqref{eq:system_2} can be transformed to a Malkin structure through a local diffeomorphism around the equilibrium point of \eqref{eq:system_2}.
\end{lemma}
The proof is provided in Appendix~I.

\begin{theorem}
\label{lemma dirscrete Malkin}
Consider the time-discretized version of Malkin structure in Definition \ref{def::malkin strucure}, where $\theta_k$ and $\rho_k$ are the $k$th sample of the quantities $\theta$ and $\rho$ in Definition \ref{def::malkin strucure}.
Then there exists a sufficiently small sampling time interval $\epsilon$ (certainly with $\epsilon<1$), and a sufficiently small open ball $\mathcal V$ around the origin such that  if $(\theta_0, \rho_0)$ lies in this open ball, then $(\theta_k, \rho_k)$ lies in the ball for all $k$ and $\rho_k \rightarrow 0$ exponentially fast.
\end{theorem}
The proof is provided in Appendix~II.
\begin{theorem}
\label{theorem_main_results}
Consider the system of equations in \eqref{formation discrete} and suppose the graph associated with the velocity measurements is connected and undirected. Then \eqref{formation discrete} is stable for sufficiently small $\epsilon_1$ and $\epsilon_2$.
\end{theorem}
\begin{proof}
To show the stability of \eqref{formation discrete} for suitable $\epsilon_i$, we initially study certain variants on \eqref{eq:system_2} and examine their stability.
First, if the second equation of \eqref{eq:system_2} is replaced  for some positive $\alpha,\beta$ by
\begin{equation}\label{eq:alphabetacts}
  \dot{\bar{v}} = \alpha L\bar{v} +\beta \bar{f}(\bar{p})
\end{equation}
the convergence properties are unaffected. Of course, the speed of convergence is changed.

Second, if \eqref{eq:alphabetacts} is replaced  for any $\epsilon>0$ by
\begin{equation}\label{eq:alphabetaepsiloncts}
\begin{split}
 \dot{\hat{p}} &=  \hat v \\
 \dot{\hat{v}} &= \epsilon \alpha L\hat{v} + \epsilon^2 \beta \bar{f}(\hat{p})
\end{split}
\end{equation}
or alternatively by
\begin{equation}\label{eq:pw}
\begin{split}
 \dot{\hat{p}} &= \epsilon\hat w \\
 \dot{\hat{w}} &= \epsilon \alpha L\hat{w} + \epsilon \beta \bar{f}(\hat{p})
\end{split}
\end{equation}
then any solution of (\ref{eq:alphabetacts}) gives rise to solutions of (\ref{eq:alphabetaepsiloncts}) and (\ref{eq:pw}) and vice versa through
\begin{equation}\label{eq:timescalechange}
\left[\begin{array}{c}
\bar p(\epsilon t)\\
\epsilon \bar v(\epsilon t)\\
\end{array}\right]
=
\left[\begin{array}{c}
\hat p(t)\\
\hat v(t)\\\end{array}\right]
=
\left[\begin{array}{c}
\hat p(t)\\
\epsilon \hat w(t)\\
\end{array}\right]
\end{equation}
The discrete-time equation with which we are working  in \eqref{formation discrete} is a discretisation of (\ref{eq:alphabetaepsiloncts}) (after some transformation and with appropriate identification of $\epsilon_1,\epsilon_2$). Now if the original equation (\ref{eq:alphabetacts}) is approximated by a difference equation with sampling interval $h$, this is equivalent to sampling (\ref{eq:pw}) with sampling interval $h/\epsilon$, or sampling (\ref{eq:alphabetaepsiloncts}) with the same sampling interval. In particular, if $h$ is such that discretisation of (\ref{eq:alphabetacts}) gives solutions which converge exponentially fast to the center manifold associated with that equation, then with discretisation interval $h/\epsilon$, solutions of the discretised version of (\ref{eq:alphabetaepsiloncts})  or (\ref{eq:pw}) will also converge exponentially fast to the center manifold. In particular,  if $\epsilon$ is chosen so that $h/\epsilon=T$, then for that value of $\epsilon$ and with the sampling interval $T$, the desired convergence will occur. In summary, if $\alpha,\beta$ are prescribed, and if a sampling interval $h$ is chosen so that the discretised version of (\ref{eq:alphabetacts}) converges to the centre manifold, then taking $\epsilon=h/T, \epsilon_1=\alpha\epsilon, \epsilon_2=\beta\epsilon^2$ will be satisfactory in \eqref{formation discrete}.
Of course, $\alpha=\beta=1$ is legitimate; with $\epsilon$ small, the values of $\epsilon_1,\epsilon_2$ will be such that velocity consensus is effectively achieved before the correct shape. This is intuitively reasonable.

So the question arises as to whether discretisation of \eqref{eq:alphabetacts} with a sufficiently small sampling interval will give convergence. Lemma \ref{lemma transfer Malkin} shows that \eqref{eq:system_2} can be transformed to a Malkin structure and therefore \eqref{eq:alphabetacts}. Furthermore, Theorem~\ref{lemma dirscrete Malkin} shows that the discretization of Malkin structure with a sufficiently small sampling interval will give convergence.
We further show in Appendix~III that the operations of coordinate basis change through a diffeomorphism to a Malkin equation and time-discretization commute.
Therefore, the theorem is proved.
\end{proof}


\subsection{Simulation Results combining Velocity Consensus and formation shape control}
Consider a three-agent system where each agent can measure it distance to the other two agents. The goal is to achieve velocity consensus and form a triangular formation. Suppose $\omega_i$ is the angular velocity of agent $i$, $T$ is the sampling time interval, $(v_{xi},v_{yi})$ is the translational velocity of agent $i$ and $(p_{xi},p_{yi})$ is the position of circle centre of agent $i$. In the simulation, we set $\omega_1=5$, $\omega_2=-3$, $\omega_3=7$, $T=2\pi$. When $t=0$, $(v_{x1},v_{y1})=(-4,1.5)$, $(v_{x2},v_{y2})=(3,-3.5)$, $(v_{x3},v_{y3})=(2,3.5)$, $(p_{x1},p_{y1})=(100,50)$, $(p_{x2},p_{y2})=(0,80)$ ,$(p_{x3},p_{y3})=(0,0)$, $\epsilon_1=5\times 10^{-2}$ and $\epsilon_2=7\times 10^{-7}$. The desired distance between each pair of agents in the formation is 20. The trajectories of the agents are shown in Fig. \ref{Formation1} and the velocities of agents are shown in Fig. \ref{Formation2}.
\begin{figure}[h]
\begin{center}
\includegraphics[width=0.35\paperwidth]{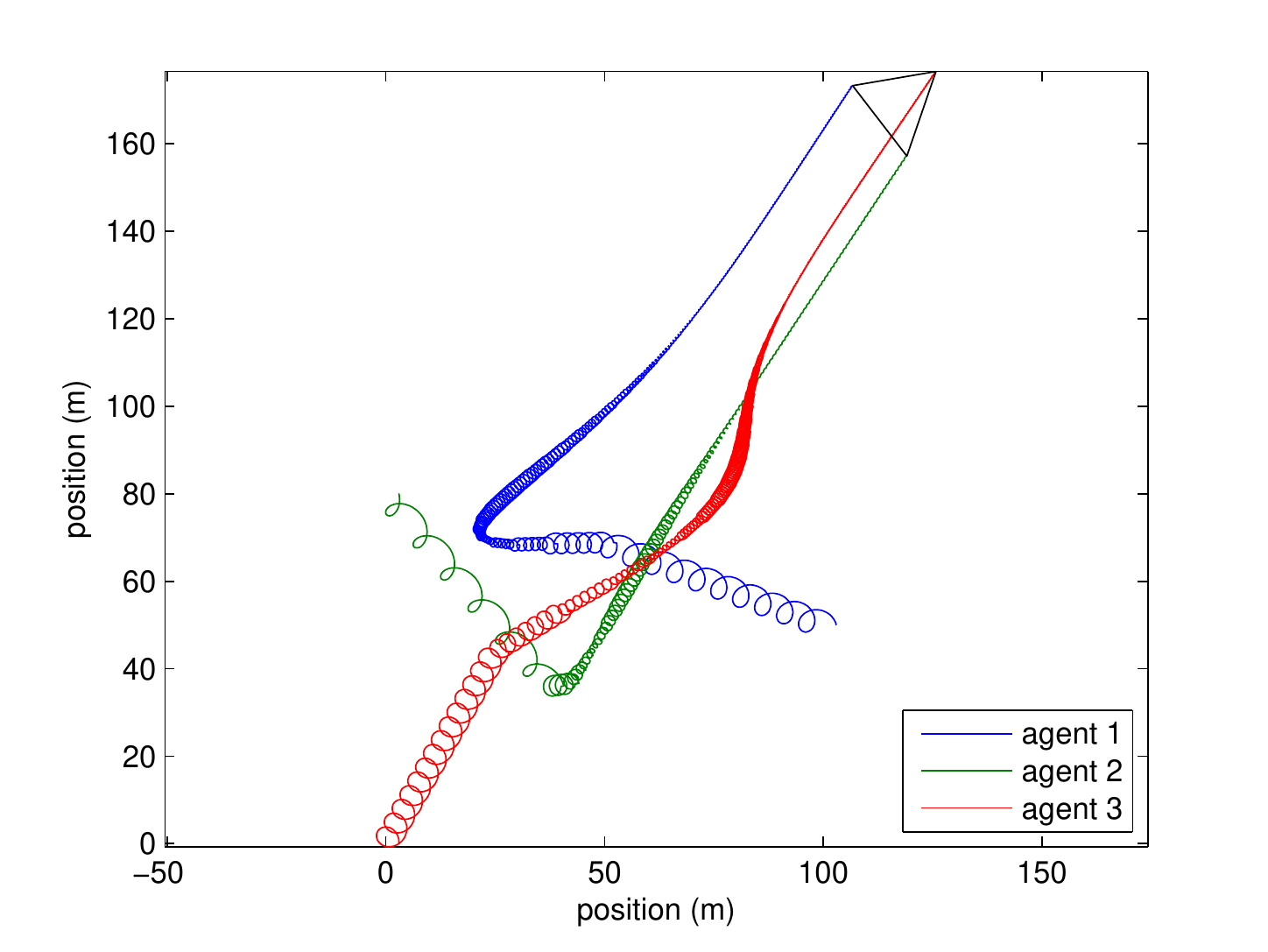}
\end{center}
\caption{Results of combining velocity consensus and formation shape control with adaptive radius setting: The trajectories of agents}
\label{Formation1}
\end{figure}

\begin{figure}[h]
\begin{center}
\includegraphics[width=0.4\paperwidth]{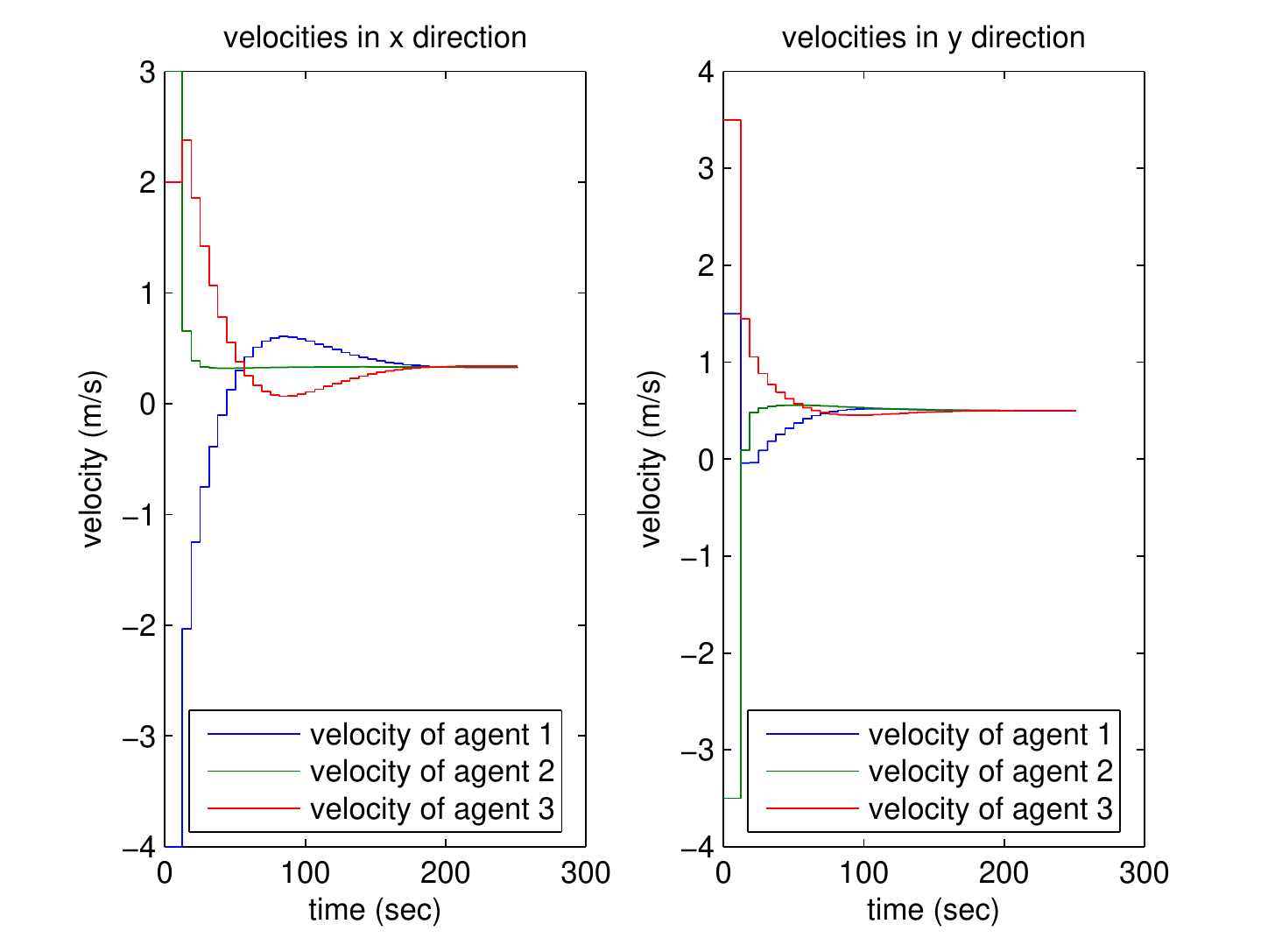}
\end{center}
\caption{Results of combining velocity consensus and formation shape control with adaptive radius setting: The translational velocities of agents}
\label{Formation2}
\end{figure}

\section{Conclusion}
\label{sec:Conclusion}
In this paper, we proposed a strategy to achieve velocity consensus and formation control using distance-only measurements for multiple agents. Given the fact that for agents to execute arbitrary motions, instantaneous distance-only measurements cannot provide enough information for achieving velocity consensus and formation control, we studied agents performing a combination of circular motion and linear motion.

In further research, we are looking to achieve formation control and velocity consensus using agents' perturbations, such that agents are not limited to perform a combination of circular motion and linear motion. In addition, it appears very likely that the same strategy as we proposed in this paper can be used in velocity consensus using bearing-only measurements.

\section*{Appendix I}
\label{proof_lemma transfer Malkin}
\textbf{Proof of Lemma \ref{lemma transfer Malkin}:}
Suppose there are $N$ agents in a formation. Consider the formation control system
\begin{equation}
\label{eq:before malkin 1}
\dot{p}=f(p)
\end{equation}
where $f(p)$ is a vector with entries
$$\sum_{j\in \mathcal{N}_i}(d^{*2}_{ij}-d^{2}_{ij})(p_i-p_j),~~~i=1.\cdots,N$$  and $d_{ij}$, $d_{ij}^*$, $p_i$ and $p_j$ are as defined in \eqref{formation conti}.
It is shown in \cite{Krick2007formation} that there is a local diffeomorphism around
the equilibrium point that transforms \eqref{eq:before malkin 1} to a Malkin structure.

Because there is a linear mapping between this $p$ and the $\bar p$ defined in \eqref{eq::transform p p bar} and the text below \eqref{eq::transform p p bar}, we know there is also a diffeomorphism
\begin{equation}
\label{eq::diffeomophism}
r=\phi(\bar p),~~~~~\bar p=\psi (r)
\end{equation}
which transforms $\dot{\bar{p}}=\bar f(\bar p)$, where $\bar f(\cdot)$ is defined in the text below \eqref{eq:system_2}, to a Malkin structure.
\begin{equation}
\label{eq:malkin_structure system}
\dot r=\left[\begin{array}{cc}
            0 & 0 \\
            0 & A
          \end{array}
\right]r+g(\theta,\rho),
~
r=\left[\begin{array}{c}
                         \theta \\
                         \rho
                       \end{array}
\right],~
g=\left[\begin{array}{c}
                         \Theta(\theta,\rho) \\
                         P(\theta,\rho)
                       \end{array}
\right]
\end{equation}
where $A$ has eigenvalues with negative real parts and $g(\theta,\rho)$ fulfills the conditions of second order term $g(\cdot)$ in Definition \ref{def::malkin strucure}. Let $n_\theta$ denote the number of elements in $\theta$ and $n_\rho$ denote the number of elements in $\rho$.

Observe
\begin{equation}\label{eq:derivative of r}
\dot{r}=\frac{\partial\phi}{\partial \bar p}\dot{\bar p}=\frac{\partial\phi}{\partial \bar p}\bar f(\bar p)
=\Big(\frac{\partial\psi}{\partial r}\Big)^{-1}{\bar f}(\psi(r)).
\end{equation}
Now the right side of \eqref{eq:malkin_structure system} and \eqref{eq:derivative of r} are the same. This is because equation \eqref{eq:malkin_structure system} is the formation control system after applying a local diffeomorphism. As a result, we obtain a further equation linking the function $\bar f(\cdot)$, $g(\cdot,\cdot)$, $\phi(\cdot)$ and the matrix $A$.
\begin{equation}\label{eq:property f}
  \frac{\partial\phi}{\partial \bar p}\bar f(\bar p)=\left[\begin{array}{cc}
            0 & 0 \\
            0 & A
          \end{array}
\right]r+g(\theta,\rho)
\end{equation}
Note the above equation reflects a property regarding $\bar f(\cdot)$ and the diffeomorphism \eqref{eq::diffeomophism}. Therefore, it is available as to draw on in considering the transformation of the second order system.

Now we are going to show that the velocity and formation shape control problem in \eqref{eq:system_2}
\begin{equation}
\label{eq:before malkin 2}
\ddot{\bar p}=L\dot{\bar p}+\bar f(\bar p),~L=L^T<0
\end{equation}
is also transformed to a Malkin structure by the same diffeomorphism stated in \eqref{eq::diffeomophism}. In the second order system, we are using the same diffeomorphism stated in \eqref{eq::diffeomophism}, but we need to  additionally know how $\dot{\bar p}$ is mapped. It is easy to obtain

\begin{equation}\label{eq:first derivative of r}
\dot{r}=\frac{\partial\phi}{\partial \bar p}\dot{\bar p}
\end{equation}
and then
\begin{equation}\label{eq:second derivative of r}
\ddot{r}=\Phi(\bar p,\dot{\bar p})+\frac{\partial\phi}{\partial \bar p}\ddot{\bar p}
\end{equation}
where the row $i$ column $j$ entry of $\Phi({\bar p},\dot{\bar p})$ takes the form $$\sum_k \frac{\partial^2\phi_i}{\partial {\bar p}_j \partial {\bar p}_k}\dot{{\bar p}}_j\dot{{\bar p}}_k.$$
Combining \eqref{eq:second derivative of r} and \eqref{eq:before malkin 2}, we obtain
\begin{equation}
\ddot{r}=\Psi(r,\dot{r})+\frac{\partial\phi}{\partial {\bar p}}L\dot{\bar p}+\frac{\partial\phi}{\partial {\bar p}}\bar f({\bar p})
\end{equation}
where $\Psi(r,\dot{r})=\Phi\left(\psi(r),\frac{\partial \psi}{\partial r} \dot{r}\right)$ is $O(\|\dot{r}\|^2)$.
Together with \eqref{eq:property f}, which as noted above remains valid, we obtain
\begin{equation}
\ddot{r}=\Psi(r,\dot{r})+\Big(\frac{\partial \psi}{\partial r}\Big)^{-1}L\Big(\frac{\partial \psi}{\partial r}\Big)\dot{r}+\left[
                                \begin{array}{cc}
                                  0 & 0 \\
                                  0 & A \\
                                \end{array}
                              \right]r+g(\theta,\rho).
\end{equation}
Now we have the system equation
\begin{equation}
\label{Malkin system equation}
\frac{d}{dt}\left[
              \begin{array}{c}
                r \\
                \dot{r} \\
              \end{array}
            \right]
=C
\left[
  \begin{array}{c}
    r \\
    \dot{r} \\
  \end{array}
\right]
+\left[
   \begin{array}{c}
     0 \\
     \Psi(r,\dot{r})+h\cdot \dot{r}+g(\theta,\rho) \\
   \end{array}
 \right]
\end{equation}
where
$$C=\left[
   \begin{array}{cc}
     0 & I \\
     \left[
        \begin{array}{cc}
        0 & 0 \\
        0 & A \\
        \end{array}
     \right] & \Big(\frac{\partial \psi}{\partial r}\Big)^{*-1}L\Big(\frac{\partial \psi}{\partial r}\Big)^* \\
   \end{array}
 \right]$$
and
$$h=-\Big(\frac{\partial \psi}{\partial r}\Big)^{*-1}L\Big(\frac{\partial \psi}{\partial r}\Big)^*+\Big(\frac{\partial \psi}{\partial r}\Big)^{-1}L\Big(\frac{\partial \psi}{\partial r}\Big)$$
where $^*$ denotes the value at system equilibrium. Observe that $C$ takes the following form
\begin{equation}\label{eq:the form of C}
C=\left[\begin{array}{cc}
  0 & D \\
  0 & E \\
\end{array}\right]
\end{equation}
with $E$ a $(n_\theta+2n_\rho)\times (n_\theta+2n_\rho)$ nonsingular square matrix.

Suppose there is a nonsingular similarity transformation $T=\left[
                                                          \begin{array}{cc}
                                                            I & -DE^{-1} \\
                                                            0 & I \\
                                                          \end{array}
                                                        \right]
$ and define $\left[
                \begin{array}{c}
                  \bar{r} \\
                  \dot{r} \\
                \end{array}
              \right]=T\left[
                \begin{array}{c}
                  r \\
                  \dot{r} \\
                \end{array}
              \right]
$. Note $\bar{r}=\left[
           \begin{array}{c}
             \bar\theta \\
             \bar\rho \\
           \end{array}
         \right]
$. There holds
\begin{equation}
\label{Malkin system transferred}
\frac{d}{dt}\left[
              \begin{array}{c}
                \bar\theta \\
                \left[
                  \begin{array}{c}
                    \bar\rho \\
                    \dot{r} \\
                  \end{array}
                \right] \\
              \end{array}
            \right]
=\left[
   \begin{array}{cc}
     0 & 0 \\
     0 & E \\
   \end{array}
 \right]
\left[
 \begin{array}{c}
   \bar\theta \\
   \left[
    \begin{array}{c}
     \bar\rho \\
     \dot{r} \\
    \end{array}
   \right] \\
 \end{array}
\right]+T\cdot o(\bar{r},\dot{r})
\end{equation}
with $\left[
   \begin{array}{cc}
     0 & 0 \\
     0 & E \\
   \end{array}
 \right]=TCT^{-1}$.

Krick points out the general conclusion in \cite{Krick2007formation} that first applying a diffeomorphism $r=\phi(\bar p)$ and then linearizing the transformed system is equivalent to first linearizing the system and then applying the diffeomorphism. While this idea was applied to the single integrator system, it remains valid for the double integrator system. It is shown in \cite{deghat_journal} that the system \eqref{eq:before malkin 2} is locally exponentially stable on a centre manifold, therefore the linearization of the system equation \eqref{eq:before malkin 2} at a point on the center manifold $\left[
                                                              \begin{array}{cc}
                                                                0 & I \\
                                                                \left(\frac{\partial {\bar f}}{\partial {\bar p}}\right)^* & L \\
                                                              \end{array}
                                                            \right]
$ has eigenvalues with non-positive real parts. Furthermore, because the local diffeomorphism is smooth, its linearization around the equilibrium
$\left[
\begin{array}{cc}
\frac{\partial \phi}{\partial p} & 0 \\
0 & \frac{\partial \phi}{\partial p} \\
\end{array}
\right]
$ is a non-singular similarity transformation. Therefore, the linearization of the system equation after applying the diffeomorphism
$$
C=\left[
\begin{array}{cc}
\frac{\partial \phi}{\partial p} & 0 \\
0 & \frac{\partial \phi}{\partial p} \\
\end{array}
\right]
\left(\left[
                                                              \begin{array}{cc}
                                                                0 & I \\
                                                                \left(\frac{\partial {\bar f}}{\partial {\bar p}}\right)^* & L \\
                                                              \end{array}
                                                            \right]\right)
\left[\begin{array}{cc}
\frac{\partial \phi}{\partial p} & 0 \\
0 & \frac{\partial \phi}{\partial p} \\
\end{array}
\right]^{-1}
$$
also has eigenvalues with non-positive real parts. With \eqref{eq:the form of C} and the fact that $A$ and $L$ are  both full rank, we know $E$ is a non-singular square matrix. Thus $E$ has eigenvalues with negative real parts.

Define $$l(\bar\theta,\bar\rho,\dot r)=o(\bar{r},\dot{r})=\Psi(r,\dot{r})+h\cdot \dot{r}+g(\theta,\rho)$$
because
\begin{enumerate}
  \item $\Psi$ is $O(\|\dot{r}\|^2)$
  \item $h=0$ when $\bar\rho=0$ and $\dot r=0$
  \item $g(\theta,\rho)$ fulfills the conditions of second order term $g(\cdot)$ in Definition \ref{def::malkin strucure}
\end{enumerate}
we can conclude (in relation to the double integrator system) that $l(\bar\theta,\bar\rho,\dot r)$ fulfills the conditions for the second order term $g(\cdot)$ of Definition \ref{def::malkin strucure}.
Therefore, we have completed the proof.
\section*{Appendix II}
\label{proof_lemma dirscrete Malkin}
\textbf{Proof of Theorem \ref{lemma dirscrete Malkin}:}

The time discretized version of Malkin structure takes the following form
\begin{equation}
\label{systems equation malkin}
\left[\begin{array}{c}
                         \theta_{k+1} \\
                         \rho_{k+1}
                       \end{array}
\right]=\left(1+\epsilon\left[\begin{array}{cc}
            0 & 0 \\
            0 & A
          \end{array}
\right]\right)\left[\begin{array}{c}
                         \theta_k \\
                         \rho_k
                       \end{array}
\right]+\epsilon \left[\begin{array}{c}
                         \Theta(\theta_k,\rho_k) \\
                         P(\theta_k,\rho_k)
                       \end{array}
\right]
\end{equation}
where $A$ has eigenvalues with negative real parts, $\Theta(\theta,0)=0$ and $P(\theta,0)=0$.
Define
$$h_1(\theta_k)=\lim_{\rho_k\rightarrow 0}\frac{\Theta(\theta_k,\rho_k)}{\|\rho_k\|},~~~~~~h_2(\theta_k)=\lim_{\rho_k\rightarrow 0}\frac{P(\theta_k,\rho_k)}{\|\rho_k\|},$$

$$b_1=\left\{\begin{array}{c}
               \frac{\Theta(\theta_k,\rho_k)}{\|\rho_k\|} ~if~\rho_k\neq 0 \\
               h_1(\theta_k)~~~ if~\rho_k=0
             \end{array}
\right.,~
b_2=\left\{\begin{array}{c}
               \frac{P(\theta_k,\rho_k)}{\|\rho_k\|} ~if~\rho_k\neq 0 \\
               h_2(\theta_k)~~~ if~\rho_k=0
             \end{array}
\right.$$

Because $\lim_{\rho\rightarrow 0}\frac{P(0,\rho)}{\|\rho\|}=0$, we know that $b_2(0)=0$. Since $A$ has eigenvalues with negative real parts, for all sufficiently small $\tau>0$, the matrix $A_d:=I+\tau A$ will have eigenvalues inside the unit circle. Without loss of generality, we may assume (using a nonsingular similarity transformation $T$ if necessary, corresponding to a replacement of $\rho_k$ by $T\rho_k$) that for some $\gamma>0$, there holds
\begin{equation}
I-A_d^{\top}A_d\geq\gamma I
\end{equation}
Now set $V(\rho_k)=\rho_k^{\top}\rho_k$.  Also, note that given any $\sigma>0$, there exists $\eta(\sigma)$  and  a closed ball $\bar {\mathcal B}_{\eta}$, without loss of generality contained in $\mathcal V$,  such that
\begin{equation}
||b_1(\rho_k,\theta_k)||\leq \sigma\;\forall (\rho_k,\theta_k)\in\bar {\mathcal B}_{\eta}
\end{equation}
Now observe that for $(\rho_k,\theta_k) \in\mathcal B_{\eta}$ there holds
\begin{eqnarray}
&&V(\rho_{k+1})-V(\rho_k)\\\nonumber
&=&\rho_k^{\top}(I-A_d^{\top}A_d)\rho_k+2\tau \rho_k^{\top}A_d^{\top}P(\rho_k,\theta_k)+\\\nonumber
&&\tau^2||P(\rho_k,\theta_k)||^2\\\nonumber
&\leq&-\gamma \rho_k^{\top}\rho_k+2\tau ||A_d||||\rho_k||^2||b_1(\rho_k,\theta_k)||\\\nonumber
&&+\tau^2||||\rho_k||^2||b_1(\rho_k,\theta_k)||^2\\\nonumber
&\leq&(-\gamma+2\tau\sigma+\tau^2\sigma^2)||\rho_k||^2\\\nonumber
&\leq&(-\gamma+2\sigma+\sigma^2)||\rho_k||^2\\\nonumber
\end{eqnarray}
Restrict $\sigma$ to be small enough that $2\sigma+\sigma^2<\gamma/2$. Then we achieve:
\begin{equation}
V(\rho_{k+1})-V(\rho_k)\leq-(\gamma/2)V(\rho_k)
\end{equation}
and
\begin{equation}\label{eq:r}
||\rho_{k+1}||^2\leq(1-(\gamma/2))||\rho_k||^2
\end{equation}
Provided that the sequence $(\rho_k,\theta_k)$ remains in ${\mathcal B}_{\eta}$, exponential convergence to zero of $\rho_k$ is achieved. We shall now argue that this can be assured through appropriate selection of the initial condition.  Suppose to obtain a contradiction that there exists a finite $K$ such that $(\rho_k,\theta_k)\in{\mathcal B}_{\eta} \forall k\in[0,K]$ but the condition fails for $k=K+1$.  Suppose that the function $b_2$, which is continuous, attains an upper bound of $\bar m$ on $\bar {\mathcal B}_{\eta}$. Observe that for all $k\in[0,K]$,
\begin{eqnarray}
||\Theta(\rho_k,\theta_k)||&=&||\rho_k||||b_2(\rho_k,\theta_k)||\\\nonumber
&\leq&\bar m||\rho_0||(1-(\gamma/2))^k
\end{eqnarray}
which implies by summation that
\begin{equation}\label{eq:s}
||\theta_{k+1}||\leq \bar m||\rho_0||\frac{1}{1-(\gamma/2)}+||\theta_0||
\end{equation}
Now restrict the initial condition $(\rho_0,\theta_0)$ to lie-in a smaller ball than $\mathcal B_{\eta}$. Define a $\eta_0<\eta$ as a positive quantity satisfying
\begin{equation}
\eta_0+\bar m\frac{1}{1-(\gamma/2)}\eta_0+\eta_0<\eta
\end{equation}
and suppose that $(\rho_0,\theta_0)\in \mathcal B_{\eta_0}$. Then while the trajectory $(\rho_k,\theta_k)$ remains in $\mathcal B_{\eta}$, ie. for all $k\in[0,K]$ with $K$ maximal, we know using (\ref{eq:r}), (\ref{eq:s}) that
\begin{eqnarray}
&&||(\rho_{k+1},\theta_{k+1})||\leq ||\rho_{k+1}||+||\theta_{k+1}||\\\nonumber
&\leq& \eta_0+\bar m\frac{1}{1-(\gamma/2)}\eta_0+\eta_0<\eta
\end{eqnarray}
This shows that $(\rho_{K+1},\theta_{K+1})\in\mathcal B_{\eta}$, and that $K$ is not maximal, i.e. there cannot be a finite $K$. Hence exponential convergence of the sequence $\rho_k$ and convergence of the sequence $\theta_k$ is established.

\section*{Appendix III}
\begin{lemma}
\label{lemma Commutativity}
Consider a differential equation $\dot p=f(p)$, with the property that a coordinate change through the diffeomorphism $r=\phi(p)$ produces a differential equation set in Malkin form. Suppose that this set is then time-discretized to obtain a discrete-time Malkin equation. Consider also the time-discretization of the equation $\dot p=f(p)$ followed by use of the diffeomorphic coordinate change $r=\phi(p)$. Then the transformed discrete-time equation is the same as that obtained as the discrete-time Malkin equation referred to above, i.e. the operations of diffeomorphic coordinate change to a Malkin equation and time-discretization commute.
\end{lemma}
\begin{proof}
Consider the Malkin structure in \eqref{eq:malkin_structure}.  Let $\psi(r)=p$ be the inverse transformation to $r=\phi(p)$. All points $\rho=0$ are equilibrium points, and therefore all points $p=\psi(\theta,0)$ are equilibrium points of the equation for $p$. It follows that $f(\psi(\theta,0))=0$. Now consider the following discretisation of the differential equation for $p$:
\begin{equation}
p_{k+1}=p_k+\epsilon f(p_k)
\end{equation}
with sufficiently small $\epsilon$. Under the mapping $r=\phi(p)$, with $J_{\phi}$ the Jacobian of $\phi(p)$, we have
\begin{equation}\label{eq:Malkind}
\begin{split}
r_{k+1}&=\phi(p_k+\epsilon f(p_k))=\phi(p_k)+\epsilon J_{\phi}f(p_k)+o(\epsilon)\\
&=r_k+\epsilon J_{\phi}f(p_k)+o(\epsilon)
\end{split}
\end{equation}
where $o(\epsilon)$ denotes higher order terms of $\epsilon$. We must show this is of a Malkin form. It is straightforward to conclude that the linear part of the discrete time equation is of a Malkin form. In order to show that the nonlinear part also has this property, what we must show is that if $\rho_k=0$, then $r_{k+1}=r_k$. This will happen if and only if the higher order terms on the right of the difference equation go to zero when $\rho_k=0.$  Accordingly, suppose $\rho_k=0$. Then we know that $p_k=\psi(\theta_k,0)$ is an equilibrium point of the differential equation for $p$, and so $f(p_k)=0$. It follows that the difference equation for which
\begin{equation}
r_{k+1}=\phi(p_k+\epsilon f(p_k))
\end{equation}
actually has $r_{k+1}=\phi(p_k+0)=\phi(p_k)=r_k$. Hence in (\ref{eq:Malkind}),  the remainder terms of higher order in $\epsilon$ all go to zero when $\rho_k$ goes to zero, therefore we have completed our proof.
\end{proof}

\bibliographystyle{IEEEtran}
\bibliography{reference}
\end{document}